\documentclass[letterpaper, 10 pt, conference]{ieeeconf}

\usepackage{latexsym}
\usepackage{amsmath}
\usepackage{amsthm}
\usepackage{graphicx}
\usepackage{amsxtra}
\usepackage{amssymb}
\usepackage{psfrag}

\IEEEoverridecommandlockouts
\newtheorem{dfn}[]{Definition}
\newtheorem{thm}{Theorem}[section]
\newtheorem{lem}[thm]{Lemma}
\newtheorem{cor}[thm]{Corollary}

\newcommand{\sat}{\mathrm{sat}}
\newcommand{\eps}{\epsilon}
\newcommand{\csi}{\xi}
\newcommand{\sect}{\mathrm{sect}}
\newcommand{\sgn}{\mathrm{sgn}}

\title{\LARGE \bf A Less Conservative Circle Criterion}
\author{D.~Materassi, M.~Salapaka, M.~Basso
        \thanks{Donatello Materassi and Michele Basso are with Dipartimento di Sistemi e Informatica, Universit\`{a} di Firenze (materassi@control.dsi.unifi.it)}
        \thanks{Murti Salapaka is with Department of Electrical and Computer Engineering, Iowa State University (murti@iastate.edu)}
        }

\begin{document}
\maketitle

\begin{abstract}
A weak form of the Circle Criterion for Lur'e systems is stated.
The result allows prove global boundedness of all system solutions.
Moreover such a result can be employed to enlarge the set of
nonlinearities for which the standard Circle Criterion 
can guarantee absolute stability.
\end{abstract}

\begin{section}{Introduction}
Although linear models can be successfully employed to locally
describe the behaviour of a physical system, they often fail to
provide a satisfactory global characterization \cite{NLDynConSys}.
An important class of nonlinear models is given by the feedback
interconnection of a SISO linear time-invariant system $G(s)$ 
and a nonlinear, possibly time-varying, static block $N$ as
represented in Figure \ref{Lur'eSystem}.
   \begin{figure}[hb]
     \begin{center}
       \psfrag{G}{ }
       \psfrag{e}{ }
       \psfrag{iwt}{ }
       \psfrag{L}{ }
       \psfrag{N}{N}
       \psfrag{(s)}{$G(s)$}
       \psfrag{x(t)}{$y(t)$}
       \includegraphics[width=0.4\textwidth]{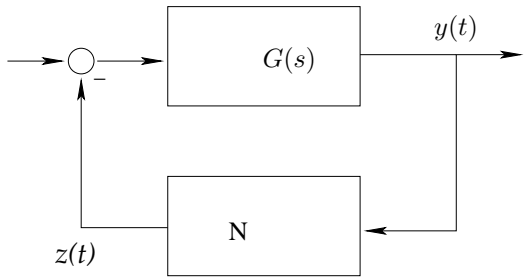}
     \end{center}
     \caption{Block diagram of a Lur'e system\label{Lur'eSystem}}
   \end{figure}
Models of this kind are known in literature as Lur'e Systems
\cite{Khalil} and their properties have been extensively studied in
the last few decades. They reveal to be a very ``expressive'' class
of models, since they can exhibit a wide variety of behaviours that
linear systems cannot show (for example the presence of isolate
equilibria, limit cycles, quasi-periodic and chaotic attractors
\cite{Genesio Quad} .\\
From this perspective (and considering also their relative
simplicity), Lur'e systems
can be considered a very attractive class to try to give a more
accurate and complete description of a physical system. 
\cite{On Describing}\\
Such a consideration shows the importance of having some
analytical tools allowing us to investigate qualitatively and
quantitatively their properties.\\
One of the significant results for Lur'e systems is  the Circle
Criterion \cite{Khalil}, \cite{Vidyasagar} which may be employed
to estabilish global asymptotical stability of the origin when the nonlinear feedback block satisfies a sector
   condition.\\
   In this paper, we first provide a relaxed form of the Circle Criterion which
   guarantees only a quantitative bound for all solutions after a finite time, but under
   conditions easier to fulfill.
   This result can be considered an useful analysis tool by itself. It can be
   employed to find quantitative bounds for any attractor in the system.\\
   % In    addition, it also admits a direct graphical interpretation.
   We also exploit this tool to develop a criterion to estabilish Absolute
   Stability for a set of nonlinearity wider than the simple
   sector class taken into account by the Standard Circle Criterion.
   The paper is organized as follows. Section \ref{Main Results} contains all the main
   theoretical results, Section \ref{A Numerical Example} gives a numerical example of
   the developed techniques and finally in Section \ref{Conclusions} main results are briefly
   summarized.
\end{section}

\begin{section}{Main Results}\label{Main Results}
   Let us consider the feedback interconnection of a SISO linear block with transfer
   function $G(s)$ and a nonlinear static block $N$ defined by a possibly time-varying
   function $n(t,\cdot)$. 
   We assume that $n(t,\cdot)$ is regular enough to guarantee the uniqueness and the
   existence of the solutions.
   Let us call $u(t)$ and $y(t)$ respectively the input and the output
   signals of the linear subsystem. Since we are considering an unforced system,
   the relation
   $u(t)=-n(t,y(t))$ holds.\\
%   We are interested in finding some conditions under which the system admits  
\noindent For sake of completeness we present some definitions.
\begin{dfn}[Positive Real Rational Function]
   A proper rational function $G(s)$ is called Positive Real if it satisfies the 
   following conditions
   \begin{itemize}
      \item all poles of $G(s)$ are in $\{s~|~Re[s]\leq 0\}$
      \item if $i\omega$ $(\omega\in\mathrm{R})$ is not a pole of $G(s)$ then
                                           $Re[G(i\omega)]\geq 0$
      \item if $i\omega$ $(\omega\in\mathrm{R})$ is a pole of $G(s)$ then it is a simple
       pole with positive residual.
   \end{itemize}
\end{dfn}

\begin{dfn}
   A proper rational function $G(s)$ is called $\epsilon$-Positive Real if
   $\epsilon>0$ and $G(s-\epsilon)$ is Positive Real.
\end{dfn}

We first introduce a simple generalization of a well-known result in passivity theory.
\begin{lem}\label{FreqGenLemma}% [\emph{B}]
Let $S$ be a Lur'e system made of a SISO linear block with transfer function $G(s)$ and  nonlinear static block $n(t,y)$.
Let us assume that
%    \begin{itemize}
%        \item $G(S)$ is SPR
%        \item $n(t,y) \in \mathrm{sect}(0,+\infty) \qquad \forall |y|> \hat y \geq 0$  
%        \item $|n(t,y) y(t)| \leq M \qquad \forall |y|\leq \hat y$.  
%    \end{itemize}
   \begin{itemize}
%       \item $\exists~\epsilon>0$: $G(s-\epsilon)$ is Positive Real 
       \item $G(s)$ is $\epsilon$-Positive Real 
       \item $ \exists~M>0:\quad-M < n(t,y) y(t)$.  
   \end{itemize}
\noindent Let $(A, B, C, D)$ be any minimal realization of $G(s)$ with a state $x$. 
Let $(P, L, J)$ be the matrices satisfing the Kalman-Yakubovich-Popov (KYP) lemma 
conditions for the realization $(A, B, C, D)$. \\
Then, for all $\eta>0$ the state $x$ of the system $S$ reaches the region 
$\Omega_{2M/\epsilon + \eta}:=\{x~|~x^T P x \leq 2M/\epsilon +\eta \}$ in a finite time.
Also this set is positively invariant.
\end{lem}

\begin{proof}
   As $(P, L, J)$ is a solution to the KYP problem, the following the relations are satisfied
   \begin{align}
       PA+A^TP &= - L^T L - \frac{\epsilon}{2} P \nonumber \\
       PB &= C^T - L^TJ \nonumber \\ 
        J&=\sqrt{2D}, \nonumber
   \end{align}
   with $P=P^T>0$.
   Consider the Lyapunov function $V(x)=\frac{1}{2}x^T P x$. We find that
   \begin{align}
     \dot V(x) &= -\frac{\epsilon}{2} V(x) 
                         -\frac{\epsilon}{2} (Lx +Ju)^T(Lx +Ju)  - n(t,y)y \leq \nonumber \\
               &\leq -\frac{\epsilon}{2} V(x)  - n(t,y)y <  M-\frac{\epsilon}{2} V(x).
   \end{align}
%   If $y>\hat y$ we have that $\dot V<0$ by the sector hypothesis.
   If $V(x)>\frac{2M}{\epsilon}$ then $\dot V$ is strictly negative.
   We use standard arguments to show that, if the initial condition $x(0)$ is
   outside $\Omega_{2M/\epsilon+\eta}$, then the state $x$ enters $\Omega_{2M/\epsilon+\eta}$
   in a finite time.
   Suppose  $x(t)$ remains outside $\Omega_{2M/\epsilon+\eta}$ for all $t$.
   Then, $V(x(t))\geq \frac{2M}{\eps} + \eta$ for all $t$.
   $V(x(t))$ is monotonically decreasing and bounded from below by $\frac{2M}{\epsilon}+\eta$.
   Thus, $V(x(t))$ converges to $\underbar V\geq \frac{2M}{\epsilon}+\eta>0$.
   Note that
   \begin{align}
        V(x(t))&=V(x(0))+\int_0^t \dot V(x(\tau)) \mathrm{d}\tau \leq\nonumber\\
	       &\leq V(x(0)) +Mt-\frac{\epsilon}{2}\underbar Vt\leq
	              V(x(0))-\frac{\epsilon~\eta}{2}t\rightarrow -\infty.\nonumber
   \end{align}
   Thus, we have a contradiction. We have concluded that for all $~\eta>0$ there exists a time
   $t_{\eta}$ such that $x(t_\eta)\in\Omega_{2M/\epsilon+\eta}$.
   It can be shown, exploiting the continuity of $V(x(t))$, that the solution $x(t)$
   cannot leave the region $\Omega_{2M/\epsilon+\eta}$ once it has entered it.
\end{proof}

% \newpage
   As a direct consequence of the previous lemma we have the following theorem 
   which can be interpreted as a weak form of the Circle Criterion.
 \begin{dfn}
   We introduce the notation $[G(s)\circlearrowright \alpha]$ to indicate
   the transfer function of the system defined by $G(s)$ with a linear static negative 
   feedback  gain $\alpha$
   \begin{equation}
       [G(s)\circlearrowright \alpha]:=\frac{G(s)}{1+\alpha G(s)}.
   \end{equation}
 \end{dfn}

\begin{thm}[Weak Circle Criterion]\label{WeakCircle}
   Let $S$ be a Lur'e system made of a SISO linear block with transfer function
   $G(s)$ and  nonlinear static block $n(t,y)$.
   Let $\alpha,\beta$ be two real scalars.
   Suppose
   \begin{equation}\label{G(s) is PR}
        [G(s)\circlearrowright \alpha] + \frac{1}{\beta-\alpha}
                                ~~is~~ \epsilon-Positive Real
   \end{equation}
   and there exists $M>0$ such that
   \begin{equation}\label{Hyp Sect Condition}
        \quad-M<\frac{1}{\beta - \alpha} (n(t,y)-\alpha y)
                                                               (\beta y - n(t,y))
   \end{equation}
%    \begin{enumerate}
%        \item $[G(s)\circlearrowright \alpha] + \frac{1}{\beta-\alpha}$ 
%                                 is $\epsilon$-Positive Real 
% 				(with same order of $G(s)$)
%        \item $ \exists~M>0:\quad-M<\frac{1}{\beta - \alpha} (n(t,y)-\alpha y)
%                                                                (\beta y - n(t,y))$.  
%    \end{enumerate}
%    \noindent (where we have introduced the notation $[[G(s)\circlearrowright \alpha]]$ 
%    to indicate
%    the transfer function of the system defined by $G(s)$ with a linear static negative 
%    feedback  gain $\alpha$).
   \noindent Let $(A, B, C, D)$ be any minimal realization of $G(s)$
    with a state $x$. 
   Let $(P, L, J)$ be a solution of the associated KYP problem associated to the matrices
   $(\hat A, \hat B, \hat C, \hat D):=(A-\alpha BC, B , C/(1+\alpha D), D/(1+\alpha D)+\frac{1}{\beta-\alpha})$. \\
   Then, for all $\eta>0$ the state $x$ of the system $S$ reaches the region 
   $\Omega_{2M/\epsilon+\eta}:=\{x~|~x^T P x < 2M/\epsilon+\eta  \}$ in a finite time
   without leaving it anymore.
\end{thm}

\begin{proof}
   Let us apply sequentially to the system $S$ an $\alpha$-poleshift and a 
   $\frac{1}{\beta-\alpha}$-zeroshift as depicted in the Figure \ref{ZeroPoleShift}.
   We obtain a new Lur'e system $\hat S$ made of a linear system 
   $\hat G:=[[G(s)\circlearrowright \alpha]] + \frac{1}{\beta-\alpha}$
   with a nonlinear feedback $\hat n(t,y):= [n(t,y)-\alpha y]$.
   It is straightfoward to note that $(\hat A, \hat B, \hat C, \hat D)$ is a
   realization of $\hat G(s)$ and also that the new output $\hat y$ is equal to
   $\frac{1}{\beta-\alpha} (\beta y - n(t,y))$. Now we can apply the Lemma
   \ref{FreqGenLemma} to conclude that there exists $t_{\eta}$ such that the state $x(t)$
   is in the region $\Omega_{2M/\epsilon+\eta}$ for all $t>t_{\eta}$.
\end{proof}~\\
\begin{figure}
   \begin{center}
     \psfrag{a}{$G(s)$}
     \psfrag{b}{$n(t,\cdot)$}
     \psfrag{c}{$\alpha$}
     \psfrag{d}[][l]{$~\frac{1}{\beta-\alpha}$}
     \psfrag{e}{$\alpha$}
     \psfrag{f}[][l]{$~\frac{1}{\beta-\alpha}$}
      \psfrag{g}{$\hat u$}
      \psfrag{h}{$u$}
      \psfrag{i}{$y$}
      \psfrag{j}{$y$}
      \psfrag{l}{$\hat y$}
      \psfrag{m}{ }
      \psfrag{n}{$u$}
      \psfrag{o}{$y$}
      \psfrag{p}{$ $}
     \includegraphics[width=0.4\textwidth]{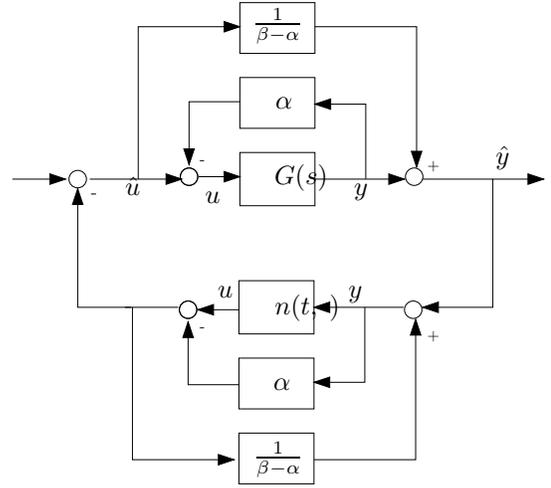}
   \end{center}
   \caption{Zero and Pole Shifting\label{ZeroPoleShift}}
\end{figure}
   The condition (\ref{G(s) is PR}) is a Circle Criterion Condition. 
   It has the following well-known characterization:
\begin{lem}\label{CircleCondition}
   There exists $\epsilon>0$ such that $[G(s)\circlearrowright \alpha] + \frac{1}{\beta-\alpha}$ is
   $\epsilon$-Positive Real if and only if,
      given on the complex plane the circle $\Gamma$ whose diameter is the segment 
      $\left[\min\left\{-\frac{1}{\alpha}, -\frac{1}{\beta} \right\},
        \max\left\{-\frac{1}{\alpha}, -\frac{1}{\beta} \right\} \right]$, we have that
   \begin{itemize}
      \item for $\alpha<\beta$ the inner part of $~\Gamma$ is contained 
             in the stability region of the Nyquist plot of $G(i\omega)$.
      \item for $\beta<\alpha$ the outer part of $~\Gamma$ is contained
             in the stability region of the Nyquist plot of $G(i\omega)$.
   \end{itemize}
\end{lem}
%  Such a characterization justifies the name of the criterion.

   The inequality (\ref{Hyp Sect Condition})
   is a condition weaker than the usual sector condition. However, 
   it is important to point out 
   that a graphical interpretation is still possible.
   For a fixed $M$, condition (\ref{Hyp Sect Condition}) is satisfied by all
   functions lying in a ``hyperbolic sector'' that contains the linear sector
   $[\alpha,\beta]$ (see Figure \ref{HyperSector}).
%   The variable $M$ can be interpreted as a parameter that tries
%   to ``enlarge'' the amplitude of the hyperbolic sector.
   \begin{figure}[hbt]
      \begin{center}
        \psfrag{y}{$y$}
        \psfrag{n(t,y)}{$n(t,y)$}
        \includegraphics[width=0.45\textwidth]{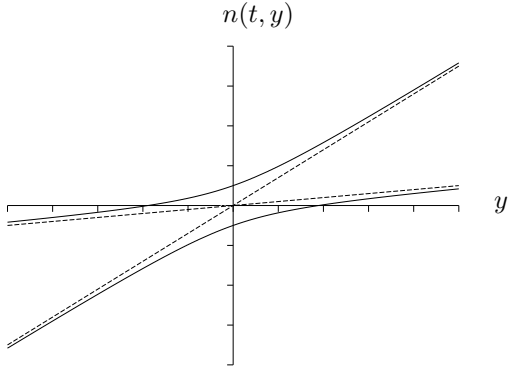}
      \end{center}
      \caption{A hyperbolic sector\label{HyperSector}}
   \end{figure}
%    It is straightforward to note that if (\ref{Hyp Sect Condition}) 
%    is verified for all $M>0$ 
%    the condition reduces to the usual sector condition and the theorem becomes
%    equivalent to the Circle Criterion.\\
   It is straightforward to note that in (\ref{Hyp Sect Condition}) if $M=0$ then 
   the condition reduces to the usual sector condition and the theorem becomes
   equivalent to the Circle Criterion.\\

   It is of interest to find a bound for the signal $|y(t)|$ when $x(t)$
   is assured to remain in the region $\Omega_{2M/\epsilon+\eta}$. 
   This can be achieved by the following result.
   \begin{lem}\label{BoundY}
      Let us consider a positive matrix $P$ that, without loss of generality, can
      be supposed symmetrical. The following optimization problem
% %      \begin{subequations}
%          \begin{align}
% 	    Y&=\max |y| \label{Y=max|y|} \nonumber \\
%             \mathrm{subject~~to} \nonumber \\
%             y&=Cx \nonumber \\
%             \frac{1}{2}x^T P x &\leq \gamma \nonumber
% 	 \end{align}
% %      \end{subequations}
      \begin{equation}
         \begin{array}{rl}
	    \max & |y| \label{Y=max|y|}\\
            \mbox{subject to} &
         \left\{
         \begin{array}{l}
            y=Cx \nonumber \\
            \frac{1}{2}x^T P x \leq \gamma
	 \end{array}
         \right.
       \end{array} 
      \end{equation}

      has a solution 
      \begin{equation}
          Y := \sqrt{2\gamma C P^{-1} C^T}
      \end{equation}
   \end{lem}
%    \begin{proof}
%        For symmetrical reason we can remove the absolute value in (\ref{Y=max|y|}) and
%        then the thesis follows trivially by the Lagrange multiplier theorem.
%  	   Fixed a level curve of the function $V(x)= x^T P x$ , the maximum value 
% 	   of $y$ is obtained for $x$ satisfying the condition
% 	   \begin{equation}
% 	     \nabla V(x)=2 P x = \lambda C^T
% 		\Rightarrow x=\frac{1}{2} P^{-1} \alpha C^T
% 	   \end{equation}
% 	   The vectors parameterized by $\lambda$ are the points where $y$ is maximum for a
% 	   fixed value of $V(x)$.
% 	   $V(x)<\frac{2M}{\epsilon}$ leads to 
% 	   $\alpha < \sqrt{\frac{2M}{4\epsilon C P^{-1} C^T}}$.
% 	   So we finally have that
% 	   \begin{equation}
% 	         y(t)\leq \sqrt{\frac{2M}{\epsilon} C P^{-1} C^T}=: Y
% 	   \end{equation}
%    \end{proof}
%    ~\\
   \noindent We have the following corollary
   \begin{cor}
      Under the hypotheses of Theorem \ref{WeakCircle} and if $G(s)$ is strictly
      proper we have that the system output is bounded by
      $Y:= 2\sqrt{\frac{M}{\eps} C P^{-1} C^T}$. Thus is
      \begin{equation}
          \forall~\eta>0~\exists~t_{\eta}: t>t_{\eta}~\Rightarrow~|y(t)|< Y+\eta.
      \end{equation}
   \end{cor}~\\
   Summarizing, the Theorem \ref{WeakCircle} gives an analytical sufficient condition to
   prove the global boundedness of all orbits for a Lur'e system after a sufficiently long
   time. A bounding region can be easily evaluated solving a KYP problem and leads to
   a bound on the output if the function $G(s)$ is strictly proper.
   In addition, the well known Circle Criterion is a special case of the derived result.\\
%   [Some notes on the fact that ``our'' KYP problem is easier than an LMI? XXX]\\
% \begin{obs}
%    The Theorem \ref{WeakCircle}, in the limit case $\alpha=0$ and $\beta=+\infty$,
%    generalizes the Lemma \ref{FreqGenLemma}.
% \end{obs}

   \noindent For a given nonlinearity, the obtained bound $Y$ on $|y(t)|$
   depends on the choice of the parameters $\alpha, \beta$ and $M$.
   Since such a choice is not unique, in order to obtain the best bound offered by the
   Theorem \ref{WeakCircle}, we can consider the following optimization problem
    \small
     \begin{equation}\label{OptBoundY}
    \hspace{-0.2cm}
       \begin{array}{rl}
            \min &Y = \sqrt{\frac{2M}{\epsilon} C P^{-1} C^T}\\
            \mbox{subject to}&
            \left\{
            \begin{array}{l}
            \frac{1}{\beta - \alpha} (n(t,y)-\alpha y)(\beta y - n(t,y))> -M\\
            P{\hat A(\alpha,\beta)}+{\hat A(\alpha,\beta)}^TP = - L^T L
	                                           - \frac{\epsilon}{2} P\\
            P {\hat B(\alpha,\beta)} = \hat C(\alpha,\beta)^T - L^TJ \\ 
            J=\sqrt{2D}.
            \end{array} \right.
       \end{array}
     \end{equation}
    \normalsize
   Unfortunately, the problem (\ref{OptBoundY}) is difficult to tackle because it
   involves the nonlinear function $n(t,\cdot)$. 
   A plausible approach is to find an upper bound to the above problem that would be 
   pursued in future works.
   The knowledge of a bound on $y$ can be used to reduce the conservativeness of
   the Circle Criterion.
   In every case, the possibility of obtaining a bound for $\max |y(t)|$ (even if 
   suboptimal or non-optimal) suggests an idea to enlarge the set of nonlinearities 
   providing absolute stability which can be normally found applying the canonical 
   Circle Criterion.
   \begin{thm}\label{TwoStepsSequentialCircleCriterion}
      Let $S$ be a Lur'e system made of a SISO linear block with a strictly proper
      transfer function
      $G(s)$ and  nonlinear static block $n(t,y)$.
      Let $\alpha_1, \beta_1, \alpha_2, \beta_2$ be four scalars
      Assume that
      \renewcommand{\labelenumi}{(\textit{\alph{enumi}})}
      \begin{enumerate}
       \item $[G(s)\circlearrowright \alpha_i] + \frac{1}{\beta_i-\alpha_i}$ 
                                is $\epsilon_i$-Positive Real~~~for i=1,2
%	\item $[G(s)\circlearrowright \alpha]_2 + \frac{1}{\beta_2-\alpha_2}$ 
%                                is $\epsilon_2$-Positive Real 
%				(with same order of $G(s)$)
       \item there exists $M_1>0$ such that 
            \begin{equation} 
              M_1<\frac{1}{\beta_1 - \alpha_1} (n(t,y)-\alpha_1 y)(\beta_1 y - n(t,y))
                  \nonumber
            \end{equation}
       \item there exists $\eta_1>0$ such that 
\begin{align}
          & |y|<2\sqrt{\frac{M_1}{\epsilon_1} C P_1^{-1} C^T}+\eta_1
                  \Rightarrow\\
           &\quad  0\leq\frac{1}{\beta_2 - \alpha_2} (n(t,y)-\alpha_2 y)
                                  (\beta_2 y - n(t,y)) \nonumber
\end{align}
%        \item if $|y|<Y_i$
%                      $0\leq\frac{1}{\beta_i - \alpha_i} (n(t,y)-\alpha_i y)
%                                       (\beta_i y - n(t,y))$.
      \end{enumerate}
      Then we can conclude Global Asymptotical Stability for the system.
   \end{thm}
   \begin{proof}
       Using (\textit{a}) and (\textit{b}) from the Corollary \ref{WeakCircle} 
       follows that there 
       exists a time $t_1$ after which the system output $y$ is such that $|y|$ will be
       bounded by $Y_1+\eta_1$ with $Y_1:=2\sqrt{\frac{M_1}{\epsilon_1} C P_1^{-1} C^T}$.
       For all $t>t_1$ we will have that the system ``explores'' the nonlinearity
       $n(t,\cdot)$
       only in the linear sector $[\alpha_2, \beta_2]$. 
       Using (\textit{c}) it follows from the standard Circle Criterion that the system is
       absolutely stable.
   \end{proof}
   ~\\
   The previous theorem can be immediately extended to cover a more general case.
   \begin{thm}\label{KStepsSequentialCircleCriterion}
      Let $S$ be a Lur'e system made of a SISO linear block with a strictly proper
      transfer function
      $G(s)$ and  nonlinear static block $n(t,y)$.
      Let $\alpha_1 ... \alpha_K, \beta_1, ... \beta_K$ be 2$K$ scalars.
      Let us assume that
      \begin{itemize}
       \item $[G(s)\circlearrowright \alpha_i] + \frac{1}{\beta_i-\alpha_i}$ 
                                is $\epsilon_i$-Positive Real~~~for $i=1,K$
%	\item $[G(s)\circlearrowright \alpha]_2 + \frac{1}{\beta_2-\alpha_2}$ 
%                                is $\epsilon_2$-Positive Real 
%				(with same order of $G(s)$)
       \item there exists $M_1>0:\quad-M_1<\frac{1}{\beta_1 - \alpha_1} (n(t,y)-\alpha_1 y)
                                  (\beta_1 y - n(t,y))$
       \item there exists $\eta_i>0: |y|<2\sqrt{\frac{M_i}{\epsilon_i} 
                                                C P_i^{-1}  C^T}+\eta_i$ 
                  $\Rightarrow -M_{i+1}\leq\frac{1}{\beta_{i+1} - \alpha_{i+1}}
		       (n(t,y)-\alpha_{i+1} y)(\beta_{i+1} y - n(t,y))~\mathrm{for~all}~
             i=1...K-2$ 
	\item there exists $\eta_{K-1}>0$ such that
           \begin{align}
            &|y|<2\sqrt{\frac{M_{K-1}}{\epsilon_{K-1}} 
                               C P_{K-1}^{-1}  C^T}+\eta_{K-1} \nonumber \\
                  &~\mbox{implies}~\nonumber\\
             &0\leq\frac{1}{\beta_{K} - \alpha_{K}}
		       (n(t,y)-\alpha_{K} y)(\beta_{K} y - n(t,y))\nonumber
           \end{align}
      \end{itemize}
      Then we can conclude Global Asymptotical Stability for the system.
   \end{thm}
   \begin{proof}
      It is a trivial extension of the previous theorem.
   \end{proof}
    
   \end{section}
   
   \begin{section}{A Numerical Example}\label{A Numerical Example}
      Let us consider the stable and minimum phase transfer function
      \begin{equation}
         G(s):=0.25\frac{(\frac{s}{5}+1)(\frac{s}{5}+1)(\frac{s}{5}+1)}
	        {(\frac{s}{20}+1)(\frac{s}{21}+1)(\frac{s}{22}+1)(\frac{s}{23}+1)}
      \end{equation}
       modeling a physical system of interest.
%       \begin{figure}[hb]
%          \begin{center}
%            \includegraphics[width=0.5\textwidth]{ExNyquist}
%          \end{center}
%          \caption{Nyquist Plot of $G(s)$\label{ExNyquist}}
%       \end{figure}
      By the Nyquist Criterion, we obtain that $G(s)\circlearrowright K$ is stable
      for every $K>0$. Suppose the task is to to design a simple proportional gain to improve
      the system performances.
      However, the actuator realizing the control law is subject to saturation.
      A noise $\xi(t)$ acting on the input of the plant is also present.
      This scenario is depicted in Figure \ref{ExBlockDiag}.
      \begin{figure}[hb]
         \begin{center}
           \psfrag{csi}{~~$\xi$}
           \includegraphics[width=0.4\textwidth]{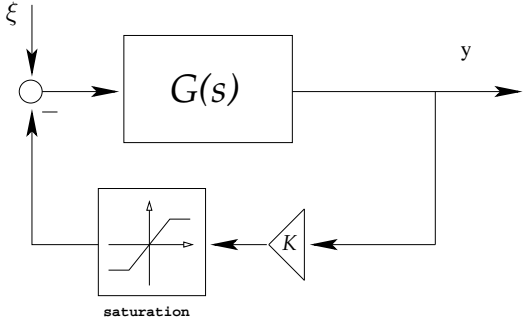}
         \end{center}
         \caption{Nyquist Plot of $G(s)$\label{ExBlockDiag}}
      \end{figure}
      We define the saturation function in the following way
      \begin{equation}
         \sat(y)=\left\{
	                 \begin{array}{l}
	                   -1~\mathrm{if}~y\leq-1\\
			   y~\mathrm{if}~|y|<1\\
			   1 ~\mathrm{if}~y\geq 1.\\
	                 \end{array}
	         \right.
      \end{equation}	
      A more sophisticated model for the controlled system is a Lur'e model $G(s)$
      and a static nonlinear feedback 
      \begin{equation}
         n_K(t,y):=S~\sat\left(\frac{Ky}{S}\right)
      \end{equation}
      where $S$ is the actuator saturation level.
      In this situation, the Nyquist Criterion can be applied only to
      estabilish local stability of the origin for any $K>0$.
      In order to obtain global results, it is possible to use the Circle
      Criterion.
      Since the nonlinearity $n_K(t,\cdot)\in\sect[0,K]$ and $G(s)$ is a stable transfer
      function, the application of the Circle Criterion to our problem is immediate.
      It can assure global asymptotical stability for all nonlinearities in
      $\sect [0,K_{cc}]$ where $K_{cc}:=-1/\min \{\mathrm{Re}{G(i\omega)} \}$.
      In our particular case we find that $K_{cc}\simeq 18.75$.
      However, for $K>K_{cc}$ the Circle Criterion does not provide any conclusion.\\
      Conversely, consider $n_K(t,y)=S\sat(Ky/S)$ for $K>K_{cc}$,  choose $\alpha=0$ and
      $\beta=3.2$. The transfer function $G(s)$ + $1/\beta$ is
      $\eps$-Positive Real with $\eps>0.8$. 
      Also note that
      \begin{equation}
         -M < \frac{1}{\beta} n_K (\beta y - n_K) < \frac{S^2}{\beta} 
	          \Rightarrow M > \frac{S^2}{\beta}.
      \end{equation}
      Thus, we can apply Theorem \ref{WeakCircle}.
      Solving the related $KYP$-problem, we find an upper bound for $y(t)$ given by
      \begin{equation}\label{ExYEst}
          y(t)\leq~Y:=2S\sqrt{\frac{CP^{-1}C^T}{\beta\eps}}\simeq 32.5~S
      \end{equation}
      which depends on the value of the parameter $S$. From (\ref{ExYEst})
      we  conclude that the system eventually enters a bounded region.
      The nonlinearity ``explored'' in this region lies in the sector
      $[\frac{S}{Y},+\infty]\simeq[0.03,+\infty]$. 
      Since this sector satisfies the Circle Condition we can conclude global asymptotical
      stability of the system for every value $K>0$.
%       In other words we have simply applied step by step the Theorem
%       \ref{TwoStepsSequentialCircleCriterion}.
      Thus we can assert that a simple 
      relay can be used as a global controller for the system.
%       So far, the effect of noise $\csi(t)$ has been neglected. 
      The developed results also allow us to take into 
      account bounded noise effects. Suppose the disturbance enters the interconnection
      in the following manner
      \begin{equation}
         n_{+\infty}(t,y):= \sgn(y)+\csi(t).
      \end{equation}
      Let us assume that $|\xi(t)|<aS$.
      In this case (\ref{Hyp Sect Condition}) is equivalent to
      \begin{align}
         -M < \frac{1}{\beta} n_{+\infty} (\beta y - n_{+\infty}) < \frac{S^2(1+a)^2}{\beta}\\
	          \mbox{thus}~~M > \frac{S^2(1+a)^2}{\beta}.
      \end{align}
      Thus, we obtain
      \begin{equation}\label{ExYEstNoise}
          y(t)\leq~Y:=S\sqrt{\frac{CP^{-1}C^T}{\beta\eps}}\simeq 32.5(1+a)S.
      \end{equation}
      The nonlinearity, in the noisy case, gets explored only in the sector 
      $[\frac{S(1-a)}{Y}, +\infty]\simeq[0.03 \frac{1-a}{1+a}, +\infty]$.
      We numerically find that for $a<0.2$ the circle condition is satisfied and 
      the system is globally asymptotically stable.
      Summarizing, we have proved that the system can be 
      successfully controlled by any linear static controller even in presence of
      saturation on the actuators. 
      Note that such an assert cannot be made using the Circle Criterion.
      In addition, it was shown that in presence of a  low  bounded 
      noise the relay controller performs a complete rejection of the disturbance.
   \end{section}

   \begin{section}{Conclusions}\label{Conclusions}
       We have presented a relaxed form of the Circle Criterion in order to prove
       a global boundedness of the trajectories of a Lur'e System. The result generalizes
       the classical Circle Criterion considering hyperbolic sectors instead of linear ones.
       This tool can be employed to extend the class of nonlinearities which the Circle
       Criterion can conclude Absolute Stability.
   \end{section}
   
\end{document}